\documentclass[11pt,A4paper,fullpage]{article}

\usepackage{todonotes} 

\usepackage{amsfonts}
\usepackage{fullpage}
\usepackage{graphicx,color}
\usepackage{xspace}
\usepackage{enumerate}
\usepackage{amssymb,amsmath}
\usepackage{amsthm}
\usepackage{subfig}

\usepackage[utf8]{inputenc}
\usepackage{times}

\usepackage[utf8]{inputenc}
\usepackage{authblk}
\usepackage{hyperref}

\usepackage{amsfonts}
\usepackage{graphicx,color}
\usepackage{xspace}
\usepackage{enumerate}
\usepackage{amssymb,amsmath}
\usepackage{amsthm}
\usepackage{tikz}

\usepackage{mathtools}

\newtheorem{theorem}{Theorem}

\newtheorem{corollary}[theorem]{Corollary}

\newtheorem{lemma}[theorem]{Lemma}

\def\multiset#1#2{\ensuremath{\left(\kern-.3em\left(\genfrac{}{}{0pt}{}{#1}{#2}\right)\kern-.3em\right)}}

\newcommand{\mch}[2]{
\left.\mathchoice
  {\left(\kern-0.48em\binom{#1}{#2}\kern-0.48em\right)}
  {\big(\kern-0.30em\binom{\smash{#1}}{\smash{#2}}\kern-0.30em\big)}
  {\left(\kern-0.30em\binom{\smash{#1}}{\smash{#2}}\kern-0.30em\right)}
  {\left(\kern-0.30em\binom{\smash{#1}}{\smash{#2}}\kern-0.30em\right)}
\right.}

\newcommand{\comp}{\ensuremath{\textrm{Cr}}\xspace}

\begin{document}

\title{Competitive Search in the Line and the Star with Predictions} %TODO Please add

\author{Spyros Angelopoulos \\ CNRS and LIP6, Sorbonne University}

\date{}

\maketitle

\begin{abstract}
We study the classic problem of searching for a hidden target in the line and the $m$-ray star, 
in a setting in which the searcher has some {\em prediction} on the hider's position. We first focus on the main metric for comparing search strategies under predictions; namely, we give positive and negative results on the {\em consistency-robustness} tradeoff, where the performance of the strategy is evaluated at extreme situations in which the prediction is either error-free, or adversarially generated, respectively. For the line, we show tight bounds concerning this tradeoff, under the {\em untrusted advice} model, in which the prediction is in the form of a $k$-bit string which encodes the responses to $k$ binary queries. For the star, we give tight, and near-tight tradeoffs in the {\em positional} and the {\em directional} models, in which the prediction is related to the position of the target within the star, and to the ray on which the target hides, respectively. Last, for all three prediction models, we show how to generalize our study to a setting in which the performance of the strategy is evaluated as a function of the searcher's desired tolerance to prediction errors, both in terms of positive and inapproximability results.
\end{abstract}

\section{Introduction}
\label{sec:introduction}

Searching for a hidden target is one of the original disciplines within the field of 
Operations Research, 
%~\cite{searchgames,alpern2013search}, 
but also a topic of significant study in Computer Science, both from the point of view of theoretical analysis and applications. This class of problems typically involves a mobile {\em searcher} that must locate an immobile {\em target} (often called {\em hider}) which hides in some unknown point of the search {\em environment}. Search problems provide natural formulations of real-life applications such as search-and-rescue missions~\cite{DBLP:journals/eor/Lidbetter20}, de-mining operations~\cite{AL:expanding},  and robot-based exploration~\cite{DBLP:journals/csr/GhoshK10}.

Among the most well-studied search problems is {\em searching on the line}, in which the environment is the unbounded line, and its generalization, the {\em $m$-ray search}, or {\em star search} problem. In the $m$-ray search problem, the environment consists of $m$ unbounded and concurrent rays, with a common point $O$, which is called the {\em origin}. Starting from $O$, the searcher must locate the target by following a {\em strategy}, defined as an infinite sequence of the form $(x_i,u_i)_{i}$, where $x_i \in \mathbb{R}^+$ and $u_i \in \{0,\ldots, m-1\}$, and with the following semantics: in iteration $i$, the searcher starts from $O$, traverses the ray $u_i$ up to distance $x_i$ from $O$, then returns back to $O$, before continuing with iteration $i+1$, until the target is eventually located. Note that for $m=2$, the star environment reduces to the infinite line.

Since the search environment is unbounded, the standard framework for evaluating the performance of a search strategy is {\em competitive analysis}, first introduced in~\cite{beck:yet.more}. Given a target $t$ hiding at some unknown point of the star, define $d(t)$ as the distance of $t$ from $O$, and $d(X,t)$ as the distance covered (or {\em cost} incurred) by a searcher that follows $X$, until $t$ is located (i.e., the first time the searcher is reached, assuming a unit-speed searcher). The  {\em competitive ratio} of $X$ is formally defined as 
\begin{equation}
\comp(X)=\sup_{t} \frac{d(X,t)}{d(t)}.
\label{eq:comp.ratio}
\end{equation}

Searching on the line has a long history of study, going back to the work of Bellman~\cite{bellman} and Beck~\cite{beck:ls}. Beck and Newman~\cite{beck:yet.more} were the first to show that an optimal competitive ratio equal to 9 can be obtained by a simple doubling strategy, i.e., a strategy of the form $x_i=2^i$. The $m$-ray search problem was first studied in the seminal works of Gal~\cite{gal:general,gal:minimax} and independently by Baeza-Yates {\em et al.}~\cite{yates:plane}. Both problems have been extended in a variety of settings and generalizations related to TCS, AI and OR since the 1960s, due to their useful abstraction of resource allocation under uncertainty. For instance, linear and ray searching have connections to the design of {\em interruptible systems} in AI~\cite{spyros:rays,steins}, the design of {\em hybrid} algorithms~\cite{hybrid}, and pipelined {\em filter ordering} in databases~\cite{Condon:2009:ADA:1497290.1497300}. They are  also involved in the analysis of strategies for more complex search problems, such as {\em spiral} search on the plane~\cite{langetepe2010optimality}. There are numerous studies on variants of linear and star search; see, e.g.,~\cite{jaillet:online,koutsoupias:fixed,HIKL99:fixed.distance,alex:robots,ray:2randomized,revisiting:esa,demaine:turn,DBLP:conf/stacs/0001DJ19,DBLP:journals/dc/CzyzowiczKKNO19,kupavskii2018lower,DBLP:conf/latin/BonatoGMP20,hyperbolic,oil,stacs-expanding,ultimate,schuierer:lower,koutsoupias:fixed,eleos} for some representative works, as well as the book~\cite{searchgames} for a game-theoretic perspective of these problems.

\subsection{Searching with predictions}
\label{subsec:intro.models}

In this work, we study the power and limitations of search strategies with {\em predictions}, in which the searcher aims to improve the competitive ratio of its strategy by leveraging some inherently imperfect information on the target. This follows a very active line of research in online computation and algorithms with incomplete information, that was initiated with the works~\cite{DBLP:conf/icml/LykourisV18} and~\cite{NIPS2018_8174}. A very large number of problems have been studied under this model (see, e.g., the survey~\cite{DBLP:books/cu/20/MitzenmacherV20} and the online collection~\cite{predictionslist}).

In regards to the search problems we study, the nature of the prediction may vary according to the application at hand. We are interested in the following models, which were introduced in~\cite{DBLP:conf/innovations/000121} in the context of linear search. 

\medskip
\noindent
{\bf The prediction is a {\em $k$-bit string}.} \ Here, the search is enhanced with a $k$-bit string that encodes some information on the target; alternatively, we may think of the prediction string as responses to binary {\em queries} given by $k$ experts. For example, a single bit can provide a (potentially erroneous) response to queries such as ``Is the target at distance at most $d$ from $O$'', or ``Is the target on an even-indexed ray?''. This is a powerful model that generalizes the concept of {\em advice complexity} so as to allow for advice that may be erroneous. Note that search and exploration problems have been studied extensively under the standard advice complexity model (see, e.g.,~\cite{dobrev2012online,fraigniaud2008tree,gorain2018deterministic,komm2015treasure,pelc2021advice}), however all such studies rely crucially on advice that is error-free. Moreover, unlike works in which each query is {\em noisy}~\cite{boczkowski2018searching}, i.e., the query responses are erroneous with some known probability, in our setting we do not rely on any probabilistic assumptions in regards to the quality of the advice.

%and algorithms with predictions, and is an application of the {\em untrusted advice} model~\cite{DBLP:conf/innovations/0001DJKR20} in the context of search games. 

\medskip
\noindent
{\bf The prediction is {\em directional}.} \ Here, the prediction is an index in $\{0, \ldots ,m-1\}$ which describes the ray on which the target lies. This is a natural prediction can be useful, for example, in a search-and-rescue application, in which there is a hint about the direction a missing person may have taken when last seen. 

\medskip
\noindent
{\bf The prediction is {\em positional}.} \ Here, the prediction describes the position of the target within the environment, namely it is of the form $(d,u)$, where $d$ corresponds to the predicted distance from the origin, and $u$ corresponds to the predicted ray on which the target hides. This is, likewise, a very natural prediction (e.g., in a search-and-rescue mission, it provides a hint about the last reported whereabouts of the missing person).

\medskip

We establish two objectives towards the evaluation of search strategies with predictions. The first objective is to find strategies of optimal, or near-optimal tradeoff between their {\em consistency} (namely, the competitive ratio assuming error-free prediction) and their {\em robustness} (namely, the competitive ratio assuming adversarially generated predictions). This is one of the standard methods of analyzing algorithms with predictions, since it establishes strong guarantees on worst-case (extreme) situations with respect to the quality of the prediction; see, e.g.,~\cite{wei2020optimal,li2021robustness,lee2021online,DBLP:conf/innovations/000121, DBLP:conf/aaai/0001K21} for applications to other online problems, and settings of incomplete information, more generally. Specifically, we are interested in showing both positive and negative results on the best-possible consistency that can be achieved by $r$-robust strategies, for any given $r$.

Our second objective goes beyond the consistency/robustness tradeoffs, and we evaluate the performance of the search strategy beyond the two extreme scenarios of error-free and adversarial error. Specifically, we study the novel setting in which the searcher defines an application-specific {\em tolerance} parameter $H$ that determines its desired tolerance to errors or, equivalently, an anticipated upper bound on the prediction error (that may be known by historical data on previous searches). This parameter is defined appropriately for each of the three prediction models we study:
Namely, in the untrusted advice model, $H$ is related to the number of erroneous advice bits 
(or query responses); in the directional model $H$ describes the distance of the predicted ray index to the one of the actual hiding target; and in the positional model, $H$ is related to the distance between the predicted and the actual target position. The tolerance model is motivated by recent works in learning-enhanced online algorithms with {\em weak predictions}, in which the prediction is an upper bound of some pertinent parameter of the input (see e.g., online knapsack with frequency predictions~\cite{knapsack:frequency}, where the prediction is an upper bound on the size of items that appear online). Our objective is thus to quantify the tradeoff between the competitive ratio and the robustness as a function of the tolerance and other parameters of the problem (e.g., the number of queries, in the query-based model). Following~\cite{knapsack:frequency}, we use the term {\em weak prediction} to refer to this setting. 

The problems we study have applications in more general decision-making settings that go beyond search. This is since $m$-ray search, as discussed above, is an abstraction of resource allocation among $m$ different tasks. To illustrate with an example, consider a researcher who has to allocate time among $m$ different projects, without knowing ahead of time which project will be completed successfully. The researcher, however, may have some intuition about which of these tasks is the most likely to succeed. This problem fits the $m$-ray search abstraction with a directional prediction. In the weak prediction setting, $H$ describes, more generally, the specific projects which the researcher believes are more likely to be successful. 

Learning-augmented search has received attention in recent years. \cite{DBLP:conf/innovations/000121} studied consistency/robustness (Pareto) tradeoffs for linear search in the three prediction models described above. \cite{DBLP:conf/innovations/BanerjeeC0L23} studied a graph search setting where every node in the graph provides a prediction of its distance to the target vertex. 
\cite{eberle2022robustification} showed how to robustify graph exploration algorithms, where the prediction is related to the spanning tree of the explored graph.

\subsection{Contribution}

Our first results apply to the untrusted advice model (i.e, the $k$-query model). We prove tight upper and lower bounds on the best consistency that an $r$-robust strategy for linear search can achieve, for any $r\geq 9$, any size of advice $k$, and with no assumptions on the nature of the strategy. This improves upon both the upper and the lower bounds of~\cite{DBLP:conf/innovations/000121}, which gave a non-tight lower bound for $k=1$ and $r=9$, and a non-tight lower bound for $r>9$ and $k=1$ for a restricted class of strategies called {\em asymptotic}. Here, the challenge is on the lower bound side. Specifically, we reduce the problem to a {\em parallel search} problem that involves $2^k$ searchers, and we rely on a novel application of Gal's functional theorem~\cite{Gal80} to prove an information-theoretic tight lower bound. While this theorem has been previously applied in parallel search problems~\cite{alex:robots}, its application in our setting is much more challenging, since we require that each of the $2^k$ searchers must be individually $r$-robust. Specifically, unlike previous works, the proof requires an explicit {\em labeling} scheme that maps the search lengths of each parallel searcher to lengths of a ``global'' sequence. We also extend our upper bound to weak predictions, by applying tools from the theory of games with a {\em lying responder}~\cite{RivestMKWS80}, in order to bound the effect of erroneous query responses to the performance.

Our second class of results is on the directional prediction model of $m$-ray search. We give the first upper and lower bounds on the consistency-robustness tradeoffs, which extend those of~\cite{DBLP:conf/innovations/000121} to star search. Here, the main challenge is again on the lower bound side, and specifically in the weak predictions setting. We show how a generalization of a biased search approach, in which the searcher allocates more time towards the predicted ray, allows us to prove an asymptotically tight bound on the competitive ratio as a function of the tolerance and the number of rays. 

Last, we show tight (Pareto-optimal) consistency-robustness tradeoffs for $m$-ray search in the positional model. As with the directional model, the only previous known results applied to linear search~\cite{DBLP:conf/innovations/000121}. The proof uses tools that circumvent the exact study of linear recurrence relations inherent in $m$-ray search problems. To our knowledge, this is a new approach towards impossibility results on this type of search games. As with the other models, we also provide tight upper and lower bounds on the competitive ratio under weak predictions. We emphasize that beyond the tight and near-tight results, the generalization to star search and the accompanied analysis under weak predictions are conceptually novel aspects of this work.

\section{Preliminaries}
\label{sec:preliminaries}

We review some notation and known results concerning $m$-ray searching. Without predictions, a
strategy is described by a sequence of the form $X=(x_i,u_i)_{i\geq 0}$; we refer to $i$ as the 
{\em iteration}\footnote{We will consider a numbering of iterations that starts either with 0, or with 1, depending on which simplifies the presentation.} of the strategy, to $x_i$ as the {\em length} of the $i$-th search {\em segment}, to $u_i$ as the ray searched in iteration $i$, and to the point at which the searcher turns in the $i$-th iteration as the corresponding {\em turn point}. Note that for linear search ($m=2$), we may assume, without loss of generality, that $u_i =i \bmod m$, 
and that $x_{i+2} > x_i$. We make the standing assumption that the target lies within distance at least a fixed value, otherwise every strategy has unbounded competitive ratio. It is well-known that the worst-case hiding positions of the target, i.e., the positions that maximize $\comp(X)$, are infinitesimally beyond the turn points of a searcher that follows $X$, namely, at distances $x_i+\epsilon$ on rays $u_i$, for $\epsilon>0$. 

% WILL USE IN LINEAR SEARCH
%
% We thus have the following equivalent expression (where $x_{-1}$ is defined to be equal to 1).
% \begin{equation}
% \comp(X)=1+2\sup_i \frac{\sum_{j=0}^{i}x_j}{x_{i-1}},
% \end{equation}
% \label{eq:cr.equivalent}

A strategy for $m$-ray search is called {\em cyclic}, if it explores the rays in a fixed permutation 
of $\{0, \ldots ,m-1\}$, e.g., if $u_i= i \bmod m$. The competitive ratio of a cyclic strategy 
of the form $X=(x_i, i\bmod m)_i$ is easily shown to be equal to
\begin{equation}
\comp(X)=1+\sup_i \frac{2\sum_{j=0}^{i+m-1}x_j}{x_i}.
\label{eq:comp.cyclic}
\end{equation}
In particular, a cyclic strategy is called {\em geometric} if $x_i=b^i$, for some fixed $b>1$ which is called the {\em base} of the strategy; we will denote such strategies by $G_b$.  Geometric strategies are significant since they are often optimal for several variants of linear search. From~\eqref{eq:comp.cyclic}, it follows that the competitive ratio of $G_b$ is therefore equal to 
$1+2b^{m}/(b-1)$. This expression is minimized for $b=m/(m-1)$, and the resulting optimal competitive ratio, denoted by $r^*_m$, is equal to 
\[
1+2\rho_m^*, \ \textrm{ where }\  \rho_m^*=\frac{m^m}{(m-1)^{m-1}}.
\]
Thus, given $r\geq r_m^*$, strategy $G_b$ has competitive ratio at most $r$ if $b^m/(b-1) \leq \rho_r$,
where $\rho_r$ is defined to be equal to $(r-1)/2$. We will denote by $b_r$ the largest such $b$ for which $G_b$ is $r$-competitive, i.e., the largest real root of the function $f(x)=x^m/(m-1)-\rho_r$. 

Under the prediction framework, the searcher is given some information $h$ in regards to the target $t$, and determines a strategy $X_h$ (we will often omit $h$ when it is clear from context). Following~\cite{DBLP:conf/innovations/000121}, we define the {\em consistency} of a strategy as its competitive ratio assuming no prediction error, and its {\em robustness} as its competitive ratio assuming adversarial prediction error. Note, in particular, that the robustness of a strategy is equal to its competitive ratio without any prediction, and we will thus use these two terms interchangeably. We say that a strategy is {\em $r$-robust} if its robustness is at most $r$ (similarly for the consistency), and that it is {\em Pareto-optimal} if its consistency and robustness are in an optimal tradeoff relation. 

Let $Y=(y_i)_{i=0}^\infty$ denote a sequence in $\mathbb{R}^+$. We define $\alpha_Y$ as
$
\alpha_Y =\limsup_{n \rightarrow \infty} y_n^{1/n}.
$
This parameter appears prominently in Gal's theorem~\cite{Gal80}, which we state below and which, informally, gives a lower bound on the supremum of a set of functionals by the supremum of these functionals over geometrically increasing sequences. From it, it follows that any $m$-ray search strategy $Y$ with search lengths $(y_i)_i$ has competitive ratio at least 
$1+2\frac{\alpha_Y^m}{\alpha_Y-1}$, hence if the strategy is $r$-competitive it must be that $\alpha_Y \leq b_r$. Given an infinite sequence 
$X=(x_0, x_1, \ldots)$, define $X^{+i}=(x_i,x_{i+1}, \ldots)$ as the subsequence starting at $i$, i.e., $X^{+i}=(x_i,x_{i+1},\ldots)$.

\begin{theorem}[\cite{Gal80}]\label{thm:sch}
  Let $X = (x_0,x_1,\ldots)$ be a sequence of positive numbers, $r$
  an integer, and $\alpha_X = \limsup_{n\rightarrow\infty} (x_n)^{1/n}$, for
$\alpha\in \mathbb{R} \cup
\{+\infty\}$.  
Let $F_i$, $i \geq 0$ be a sequence of functionals which satisfy the following properties:
  \begin{enumerate}
  \item $F_i(X)$ only depends on $x_0,x_1, \ldots ,x_{i+r}$,
    \item $F_i(X)$ is continuous for all $x_k>0$, with $0 \leq k
      \leq i + r$, 
    \item $F_i(\lambda X) = F_i(X)$, for all $\lambda > 0$, 
    \item $F_i(X+Y) \leq \max(F_i(X),F_i(Y))$, and
    \item $F_{i+k}(X) \geq F_i(X^{+k})$, for all $k \geq
      1$,  
  \end{enumerate}
  then
  \[
     \sup_{0 \leq i < \infty} F_i(X)
     \geq  \sup_{0 \leq i < \infty} F_i(G_{\alpha_X}), \ 
     \quad \textrm{where $G_a=(a^0,a^1,a^2, \ldots )$.}
  \]
\end{theorem}

\section{Linear search with untrusted advice}

In this section, we study linear search in a model in which the prediction is an untrusted advice string of size $k$. 
We first show optimal upper and lower bounds on the best consistency of $r$-robust strategies, then in Section~\ref{subsec:advice.noisy} we study the extension to weak predictions. 

Our results will show and exploit connections between a single-searcher strategy with $k$-bit advice, and a {\em multi-searcher} strategy with $2^k$ parallel searchers, but no advice.  Hence, we first present some definitions and notation concerning the setting of $p>1$ parallel searchers, labeled from the set $\{0, \ldots ,p-1\}$. In a $p$-searcher strategy, each searcher $j$ defines its own strategy of the form $X_j=(x_{j,i},u_{j,i})_{i=0}^\infty$. 
We thus denote the $p$-searcher strategy as ${\cal X}=\{X_j\}_{j=0}^{p-1}$, or equivalently, we say 
that it is defined by the set $\{X_j\}_{j=0}^{p-1}$. 
The competitive ratio of a $p$-searcher strategy is the worst-case ratio of the first time one of the $p$ searchers finds the target $t$ (assuming unit-speed searchers) and the distance $d(t)$ of the target from the origin~\cite{alex:robots}. 

Observe that the optimal consistency of an $r$-robust strategy with $k$ advice bits is equal to the competitive ratio of a parallel search strategy that is defined by $2^k$ searchers, each of which is individually $r$-robust. Namely, if the advice is error-free, it can be used to select the single-searcher strategy, among the $2^k$ ones, that reaches the target at optimal cost. Note that, by construction, the robustness of this strategy is at most $r$, since each individual searcher is $r$-robust. This observation applies to both positive and negative results on the consistency/robustness tradeoffs.

% \subsection{Upper bound}
% \label{subsec:advice.upper}

We first show an upper bound on the consistency of $r$-robust strategies:
\begin{theorem}
For any $r\geq 9$, there is an $r$-robust strategy for searching on the line with $k$-bit advice that has consistency at most $1+2\frac{b_r^{1/q}}{b_r-1}$, where $q=2^{k-1}$.
\label{thm:upper.advice}
\end{theorem}

\begin{proof}
Let ${\cal S}$ denote the $2^k$-parallel strategy as defined by the set 
$S_0, \ldots S_{2^k-1}$ where
\[
S_j =(b^{j+iq}, i \bmod 2), \ \textrm{if $j$ is even  and } \
S_j =(b^{j+iq}, (i+1) \bmod 2), \ \textrm{if $j$ is odd},
\]
for some $b>1$ that will be specified later. That is, each individual strategy is near-geometric, and half of the searchers explore ray 0 in their first iteration, whereas the other half explore ray 1.  We require that each strategy in ${\cal S}$ is $r$-robust which implies, from the discussion in Section~\ref{sec:preliminaries}, that $b$ must satisfy $b^q \leq b_r$, hence $b\leq b_r^{1/q}$. 

For any given target $t$, let $C_t$ denote the smallest cost at which one of the $2^k$ strategies in ${\cal S}$ discovers $t$, namely $C_t =\min_{j} d(S_j,t)$. 
% It follows that we obtain a strategy of consistency $\sup_t C_t/d(t)$, since the $k$ bits can be used to identify the strategy that minimizes 
% $C_T$ in {\cal S}$. $
Consider a target placed at distance $d(t)  \in(b^i,b^{i+1}]$, say on ray 1 (if the target is on ray 0, the argument is symmetric). Let $\tau$ and $l$ be defined such that 
\[
l =\lfloor \frac{i-q+1}{q} \rfloor, \ \textrm{ and } \ \tau= i-q+1 \bmod q,
\]
hence $i-q+1=lq+\tau$. By construction, there is a strategy in ${\cal S}$ which finds $t$ at cost at most
\begin{align*}
2\sum_{m=0}^l b^{\tau+mq}+d(t)&=2\frac{b^{\tau+(l+1)q}-b^\tau}{b^q-1}+d(t) \\
& \leq 2\frac{b^{\tau+(l+1)q}}{b^q-1}+d(t)\\
&=2\frac{b^{i-q+1+q}}{b^q-1}+d(t)\\ &=2\frac{b^{i+1}}{b^n-1}+d(t). 
\end{align*}
Thus
\begin{align*}
\frac{C_t}{d(t)}&=1+2\frac{b^{i+1}}{(b^q-1)d(t)}\\
&\leq 1+2\frac{b^{i+1}}{(b^q-1)b^i}=1+2\frac{b}{b^q-1}.
\end{align*}
Last, note that the function $b/(b^q-1)$ is decreasing in $b$. Choosing $b = b_r^{1/q}$ optimizes the competitive ratio of the strategy and yields the result. 
\end{proof}

Note that this upper bound is not only of theoretical significance, but can be obtained in practice via a query-based implementation. This is because the $i$-th advice bit can be interpreted, equivalently, as a response to a subset query that asks whether the target is hiding within a specific subset of the infinite line. Informally, the theorem shows which questions to ask to $k$ different experts\footnote{Experts may be inherently erroneous; in Theorem~\ref{thm:advice.noisy} we extend the result to account for query errors.} about the whereabouts of the target so as to maximize the efficiency of search, while remaining robust to adversarial responses.

We now move to the lower bound. We first show a useful property of parallel search. To illustrate the property, consider Figure~\ref{fig:snapshot}, which shows the first segments of a $p$-parallel strategy defined by $S_1, \ldots S_p$. In this example, the first segment of strategies $S_1, \ldots S_i$ is to the left ray of the line, whereas the first segment of $S_{i+1}, \ldots ,S_p$ is to the right ray of the line. Furthermore, the lengths of these segments are in increasing order, in the left and the right ray, respectively, as illustrated. We observe that, without loss of generality, a target that hides infinitesimally beyond the first turnpoint in $S_j$, with 
$j \in \{1,\ldots, p-1\}$ is first discovered by $S_{j+1}$, and if $j=p$, it is first discovered by $S_1$. This is because, if this was not the case, then one of the strategies would mark its second turn before it had explored any new parts of the line, which would mean that the corresponding second segment would be redundant and thus could be omitted.

\newcommand{\snap}{

\tikzset{every picture/.style={line width=0.75pt}} %set default line width to 0.75pt        

\begin{tikzpicture}[x=0.75pt,y=0.75pt,yscale=-1,xscale=1]

%Straight Lines [id:da44391897337497765] 
\draw    (32,56) -- (634,55) ;
%Flowchart: Connector [id:dp12599016759429615] 
\draw  [fill={rgb, 255:red, 0; green, 0; blue, 0 }  ,fill opacity=1 ] (322,56.5) .. controls (322,53.46) and (324.69,51) .. (328,51) .. controls (331.31,51) and (334,53.46) .. (334,56.5) .. controls (334,59.54) and (331.31,62) .. (328,62) .. controls (324.69,62) and (322,59.54) .. (322,56.5) -- cycle ;
%Straight Lines [id:da2110828193074139] 
\draw    (277,76) -- (325,76) ;
\draw [shift={(328,76)}, rotate = 180] [fill={rgb, 255:red, 0; green, 0; blue, 0 }  ][line width=0.08]  [draw opacity=0] (8.93,-4.29) -- (0,0) -- (8.93,4.29) -- cycle    ;
\draw [shift={(274,76)}, rotate = 0] [fill={rgb, 255:red, 0; green, 0; blue, 0 }  ][line width=0.08]  [draw opacity=0] (10.72,-5.15) -- (0,0) -- (10.72,5.15) -- (7.12,0) -- cycle    ;
%Straight Lines [id:da29555779012631556] 
\draw    (234,116.98) -- (325,116.52) ;
\draw [shift={(328,116.5)}, rotate = 179.7] [fill={rgb, 255:red, 0; green, 0; blue, 0 }  ][line width=0.08]  [draw opacity=0] (8.93,-4.29) -- (0,0) -- (8.93,4.29) -- cycle    ;
\draw [shift={(231,117)}, rotate = 359.7] [fill={rgb, 255:red, 0; green, 0; blue, 0 }  ][line width=0.08]  [draw opacity=0] (10.72,-5.15) -- (0,0) -- (10.72,5.15) -- (7.12,0) -- cycle    ;
%Straight Lines [id:da1681088692231547] 
\draw    (186.5,189.5) -- (321,189.5) ;
\draw [shift={(324,189.5)}, rotate = 180] [fill={rgb, 255:red, 0; green, 0; blue, 0 }  ][line width=0.08]  [draw opacity=0] (8.93,-4.29) -- (0,0) -- (8.93,4.29) -- cycle    ;
\draw [shift={(183.5,189.5)}, rotate = 0] [fill={rgb, 255:red, 0; green, 0; blue, 0 }  ][line width=0.08]  [draw opacity=0] (10.72,-5.15) -- (0,0) -- (10.72,5.15) -- (7.12,0) -- cycle    ;
%Straight Lines [id:da10611896123780229] 
\draw    (331.5,76.46) -- (397,75.54) ;
\draw [shift={(400,75.5)}, rotate = 179.2] [fill={rgb, 255:red, 0; green, 0; blue, 0 }  ][line width=0.08]  [draw opacity=0] (8.93,-4.29) -- (0,0) -- (8.93,4.29) -- cycle    ;
\draw [shift={(328.5,76.5)}, rotate = 359.2] [fill={rgb, 255:red, 0; green, 0; blue, 0 }  ][line width=0.08]  [draw opacity=0] (10.72,-5.15) -- (0,0) -- (10.72,5.15) -- (7.12,0) -- cycle    ;
%Straight Lines [id:da12934394421035544] 
\draw    (331.5,117) -- (412,117) ;
\draw [shift={(415,117)}, rotate = 180] [fill={rgb, 255:red, 0; green, 0; blue, 0 }  ][line width=0.08]  [draw opacity=0] (8.93,-4.29) -- (0,0) -- (8.93,4.29) -- cycle    ;
\draw [shift={(328.5,117)}, rotate = 0] [fill={rgb, 255:red, 0; green, 0; blue, 0 }  ][line width=0.08]  [draw opacity=0] (10.72,-5.15) -- (0,0) -- (10.72,5.15) -- (7.12,0) -- cycle    ;
%Straight Lines [id:da5901208059665481] 
\draw    (333.5,165) -- (481,165) ;
\draw [shift={(484,165)}, rotate = 180] [fill={rgb, 255:red, 0; green, 0; blue, 0 }  ][line width=0.08]  [draw opacity=0] (8.93,-4.29) -- (0,0) -- (8.93,4.29) -- cycle    ;
\draw [shift={(330.5,165)}, rotate = 0] [fill={rgb, 255:red, 0; green, 0; blue, 0 }  ][line width=0.08]  [draw opacity=0] (10.72,-5.15) -- (0,0) -- (10.72,5.15) -- (7.12,0) -- cycle    ;

% Text Node
\draw (245,63) node [anchor=north west][inner sep=0.75pt]   [align=left] {$\displaystyle S_{1}$};
% Text Node
\draw (205.5,105) node [anchor=north west][inner sep=0.75pt]   [align=left] {$\displaystyle S_{2}$};
% Text Node
\draw (169.45,187.19) node   [align=left] {$\displaystyle S_{i}$};
% Text Node
\draw (420.45,75.19) node   [align=left] {$\displaystyle S_{i+1}$};
% Text Node
\draw (440.45,115.19) node   [align=left] {$\displaystyle S_{i+2}$};
% Text Node
\draw (500.45,164.69) node   [align=left] {$\displaystyle S_{p}$};
% Text Node
\draw (327.95,28.19) node   [align=left] {$\displaystyle O$};

\end{tikzpicture}

}
\begin{figure}
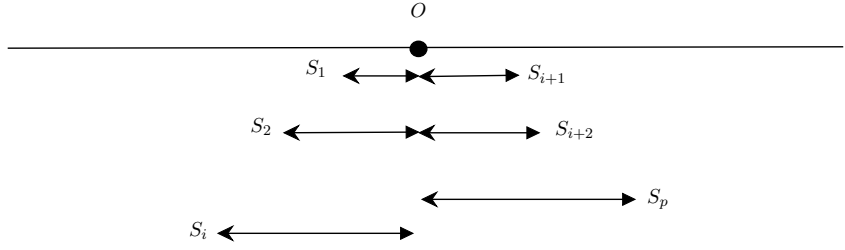

\centering
\scalebox{0.7}{\snap}
\caption{A snapshot of the first iteration in a $p$-parallel strategy.}
\label{fig:snapshot}
\end{figure}

We can argue, inductively, that the same property extends not only to targets hiding infinitesimally beyond the first turn points, but beyond {\em every} turn point. To formalize this concept, let ${\cal S}$ be a $p$-parallel strategy defined by single-searcher strategies $S_1, \ldots ,S_p$. We say that $S_j$ is {\em responsible for the $i$-th turn point of $S_l$} if $S_j$ is the first strategy
in ${\cal S}$ to discover a target hiding infinitesimally beyond the $i$-th turnpoint of $S_l$. 
The following lemma shows that it suffices to consider 
$p$-parallel strategies in which a ``snapshot'' of the first iteration of the $p$ individual strategies
provides a global picture about the relative turnpoints of all individual strategies, for all subsequent iterations. This will help us setup the lower bound. 

\begin{lemma}
For any $p$-searcher strategy ${\cal S}$, there is a $p$-searcher strategy ${\cal S'}=\{S'_1, \ldots ,S'_p\}$ such that there is a bijection $\pi:\{1, \ldots p \} \to \{1, \ldots p\}$ with the property that for any $j \in \{1, \ldots ,p\}$, $S'_{\pi(j)}$ is  responsible for the $i$-th turn point of $S'_j$ for all $i \in \mathbb{N^+}$, and ${\cal S'}$ has competitive ratio no worse than 
${\cal S}$. 
\label{lemma:property.first.order}
\end{lemma}

\begin{proof}
We prove the statement by induction on the time of {\em configurations}. Here, a configuration at time $i$ consists of the pairs $(x_{j,i}, u_{j,i})$, where $x_{j,i}$ is the length of the $i$-th segment of $S_j'$, and $u_{j,i}$ the ray on which the $i$-th iteration takes place. We will prove a stronger statement, namely that for all $i$, the relative ordering of the lengths $(x_{j,i})$ is the same, for all $j$ such that $S_j'$ search the same ray at iteration $i$. 
The statement holds for $i=1$, as explained earlier. Suppose that the statement is true for $i=l$, and that there is a bijection $\pi$ with the desired property. We can then apply the same argument on the configuration at time $l+1$: this is because if $S'_j$ is responsible for the $l$-th turnpoint of $S'_{\pi(j)}$, then it will be the first to explore a point infinitesimaly beyond the $(l+1)$-th turnpoint, otherwise $S_j$ would not explore any new parts of the line in iteration $l+1$, a contradiction (since we can assume, without loss of generality, that every strategy satisfies such property).
\end{proof}

We now show how to prove the lower bound.
\begin{theorem}
For any $r\geq 9$, every $r$-robust strategy for searching on the line with untrusted advice 
of size $k$ has consistency at most $1+2\frac{b_r^{1/q}}{b_r-1}$, where $q=2^{k-1}$.
\label{thm:advice.lower}
\end{theorem}

\begin{proof}
For convenience of notation, let $n=2^k$, and let ${\cal S}$ be an $n$-parallel strategy defined by 
$S_1, \ldots S_{n}$, and which satisfies the property of Lemma~\ref{lemma:property.first.order}. Let $i$ be such that strategies $S_1, \ldots S_i$ start their first iteration to the left (in increasing order of this length), whereas $S_{i+1}, \ldots S_{n}$ start their first iteration to the right (again in increasing length order). Thus, we have that $S_j$ is responsible for $S_{j-1}$, for all $j \in \{2, \ldots ,2^k\}$, whereas $S_1$ is responsible for $S_n$. The situation is depicted in Figure~\ref{fig:snapshot}, where $p=n$. 

Let $s_{j,m}$ denote the search length of the $m$-th iteration of $S_j$. For any fixed $m$, consider a target hiding infinitesimally beyond the $m$-th turn point of $S_j$, for each $j$. 
Since $S_{j+1}$ is responsible for $S_j$, for all $j \in \{1, \ldots i-1, i+1, \ldots n-1 \}$, we have
\begin{equation}
\comp({\cal S}) \geq 1+2\frac{\sum_{l=1}^{m-1} s_{j+1,l}}{s_{j,m}},
\ \text{ for all $j \in \{1, \ldots i-1, i+1, \ldots n-1 \}$}. 
\label{eq:lower.1st.A}
\end{equation}
In addition, since $S_{i+1}$ and $S_1$ are responsible for $S_i$ and $S_{n}$, respectively, we have that
\begin{equation}
\comp({\cal S}) \geq 1+2\frac{\sum_{l=1}^{m-1}s_{i+1,l}}{s_{i,m-1}} \quad \text{ and} \  
\comp({\cal S}) \geq 1+2\frac{\sum_{l=1}^{m-1}s_{1,l}}{s_{n,m-1}},
\label{eq:lower.1st.B}
\end{equation}
and note the subtle, but important differences in the indexing of the denominators between~\eqref{eq:lower.1st.A} and~\eqref{eq:lower.1st.B}. This motivates our next step, in which we  {\em label} the lengths of all search segments of the $2^k$ strategies in ${\cal S}$ in a way that will allow us to use the above lower bounds.
%, in combination with Theorem~\ref{thm:sch}.  
Let $\{x_l\}_{l=1}^\infty$ denote the set of all segment lengths in the parallel strategy ${\cal S}$. We map bijectively each such length to a segment length of one of the strategies $S_j$, according to the following function.
\begin{equation}
s_{j,m}=
\begin{cases}
x_{j+mn}, &\text{if $j \in [1,\ldots i-1]$} \\
x_{j-1+mn} & \text{if $j \in [i+1,\ldots n-1]$}\\
x_{n-1+mn}, &\text{if $j=i$} \\
x_{n+mn}, &\text{if $j=n$}.
\end{cases}
\label{eq:labeling}
\end{equation}
Using the property $\max\{\frac{a}{b},\frac{c}{d}\} \geq \frac{a+b}{c+d}$, we obtain 
from~\eqref{eq:lower.1st.A} and~\eqref{eq:lower.1st.B}:
\begin{equation}
\comp({\cal S}) \geq 1+2\frac{\sum_{j=1}^{i-1}\sum_{l=1}^{m-1} s_{j+1,l}+\sum_{l=1}^{m-1}s_{i+1,l}+
\sum_{j=i+1}^{n-1} \sum_{l=1}^{m-1} s_{j+1,l}+ \sum_{l=1}^{m-1}s_{1,l}}
{\sum_{j=1}^{i-1}s_{j,m}+s_{i,m-1}+\sum_{j=i+1}^{n-1}s_{j,m}+s_{n,m-1}}.
\label{eq:lower.1st.C}
\end{equation}
Using the labeling of~\eqref{eq:labeling}, the numerator in~\eqref{eq:lower.1st.C} can be written equivalently as
\[
\sum_{j=1}^{i-2} \sum_{l=1}^{m-1} x_{j+1+ln} +
\sum_{l=1}^{m-1} x_{n-1+ln} +
\sum_{l=1}^{m-1} x_{i+ln} +
\sum_{j=i+1}^{n-2} \sum_{l=1}^{m-1} x_{j+ln} +
\sum_{l=1}^{m-1} x_{n+ln},
\]
which, after gathering terms, is equal to $\sum_{l=1}^{mn}$. Moreover, 
the denominator in~\eqref{eq:lower.1st.C} is equivalently  
\[
\sum_{j=1}^{i-1} x_{j+mn} + x_{n-1+(m-1)n}+ \sum_{j=i+1}^{n-1} x_{j-1+mn}+x_{n+(m-1)n}=
\sum_{l=0}^{n-1}x_{mn-1+l}.
\] 
We thus obtain
\begin{equation}
\comp({\cal S}) \geq 1+2\sup_m \frac{\sum_{l=1}^{mn} x_l}{\sum_{l=0}^{n-1}x_{mn-1+l}}.
\label{eq:upper.advice.nice}
\end{equation}
Define the functional $F_m$ to be $F_m=\frac{\sum_{l=1}^{mn} x_l}{\sum_{l=0}^{n-1}x_{mn-1+l}}$. It is easy to see that this functional satisfies the conditions of Gal's Theorem~\cite{Gal80}. Therefore,
\begin{equation}
\comp({\cal S}) \geq 1+2\sup_m \frac{\sum_{l=1}^{mn} \alpha^l}{\sum_{l=0}^{n-1}\alpha^{mn-1+l}},
\ \textrm{ where $\alpha=\limsup_{l \to \infty} x_l^{1/l}$}.
\label{eq:lower.1st.D}
\end{equation}
If $\alpha\leq 1$, the RHS of~\eqref{eq:lower.1st.D} is unbounded. If $\alpha>1$,~\eqref{eq:lower.1st.D} gives
\begin{equation}
\comp({\cal S}) \geq 1+2 \sup_m \frac{\alpha^{mn+1}-\alpha}{(\alpha-1)\alpha^{mn-1}\frac{\alpha^n-1}{\alpha-1}}\geq 1+2\frac{\alpha^2}{\alpha^n-1}.
\label{eq:last.advice}
\end{equation}
We will now use the fact that each strategy in ${\cal S}$ is individually $r$-robust, in order to bound $\alpha$ from below, and thus $\comp({\cal S})$ as well. From~\cite{jaillet:online} we know that any $r$-competitive single searcher strategy of the form $Y=(y_j)_{j=1}^\infty$ satisfies 
$y_j= O(b_r^j)$. Given the labeling scheme~\eqref{eq:labeling}, it follows that $x_j=O(b_r^{j/n})$,
hence from the definition of $\alpha$, $\alpha \leq b_r^{1/n}$. Therefore,~\eqref{eq:last.advice} gives
%$\alpha \leq b_r^{1/n}$. Since the function $x^2/(x^n-1)$ is decreasing for $x>1$, we have that
\[
\comp({\cal S}) \geq 1+2\frac{b_r^{2/n}}{b_r-1} = 1+2\frac{b_r^{1/q}}{b_r-1}, 
\]
which completes the proof.
\end{proof}

We give some intuition behind the significance of the mapping~\eqref{eq:labeling} in the proof. The labeling accomplishes two goals: First, it leads to~\eqref{eq:upper.advice.nice}, whose sums in the numerator and the denominator contain summands with ``contiguous'' indices: this is an essential requirement for the application of Gal's theorem. Second, it implies that each strategy $S_j$
is of the form $(x_{\pi(j)+nl})_{l=0}^\infty$, where $\pi$ is a bijection over $\{1, \ldots ,n\}$, which allows us to argue that $\alpha \leq b_r^{1/n}$.

\subsection{Extension to weak predictions}
\label{subsec:advice.noisy}

We show how to extend the upper bound to incorporate weak predictions. In this setting, as discussed in Section~\ref{sec:introduction}, given advice of size $k$, and robustness requirement $r\geq 9$, the searcher specifies a tolerance parameter $H\leq k/2$. The objective is to obtain an $r$-robust strategy of minimum competitive ratio assuming that at most $H$ advice bits are erroneous.   

To address this problem, we will make use of a result by Rivest {\em et al.}~\cite{RivestMKWS80}, who
studied games with a {\em lying responder}. In their setting, given $k \in \mathbb{N}^+$, 
$H\leq k/2$, and a domain ${\cal D}=\{1,\ldots,m\}$, the objective is to find the index of an unknown $x \in {\cal D}$, using $k$ queries, of which up to $H$ may receive incorrect responses. A query can be a comparison query of the form ``Is $x\leq M$?'', for some given $M \in [1,m]$, or more generally, a subset query of the form ``Is $x$ in $S$?'', where $S$ is a subset of the domain ${\cal D}$. Define the sum of binomial coefficients
$
\multiset{N}{m}   := \sum_{j=0}^m {N \choose j}, \  \textrm{for $m \leq N$.}
$
In \cite{RivestMKWS80} it was shown that as long as $m \leq 2^k/\multiset{k}{H}$, $k$ comparison queries suffice to find $x$ in the above game, in the presence of at most $H\leq k/2$ errors. This leads to the following extension of Theorem~\ref{thm:upper.advice}.

\begin{theorem}
For any $r\geq 9$, and any $H \leq k/2$, there is an $r$-robust strategy for searching on the line with $k$-bit advice that has competitive ratio at most $1+2\frac{b_r^{1/q}}{b_r-1}$, where $q={2^{k-1}}/{\multiset{k}{H}}$.
\label{thm:advice.noisy}
\end{theorem}

\begin{proof}
Define a $(2q)$-searcher parallel strategy ${\cal S}$, determined by the set $\{S_0, S_1, \ldots S_{2q-1}\}$, where each $S_j$ is as in the proof of Theorem~\ref{thm:upper.advice}. Given ${\cal S}$, the advice dictates which strategy to choose. More precisely, as long as $2q \leq \frac{2^{k}}{\multiset{k}{H}}$, the result of~\cite{RivestMKWS80} guarantees that this is indeed possible.
% : Using $2H+1$ bits (i.e., comparison queries), we can determine on which ray the target lies. The remaining $k-2H-1$ bits can be used to identify the strategy that finds the target at minimum cost. More precisely, as long as $q \leq \frac{2^{k-2H-1}}{\multiset{k-2H-1}{H}}$, the result of~\cite{RivestMKWS80} guarantees that this is indeed possible. 
The proof proceeds along the lines of Theorem~\ref{thm:upper.advice}, by setting $q=\frac{2^{k-1}}{\multiset{k}{H}}$.
\end{proof}

%Using entropy-based approximations of the partial sum of binomial coefficients (see Appendix for details), 
We can further show that
$
\frac{1}{q} \leq \frac{1}{2^{k(1-{\cal H}(\frac{H}{k}))-1}},
%\frac{1}{2^{(k-2H-1)(1-{\cal H}(\frac{H}{k-2H-1}))}},
$  
where ${\cal H}$ denotes the {\em binary entropy} function. This allows for a direct comparison to the Pareto-optimal upper bound of Theorem~\ref{thm:upper.advice}. In particular, we observe that as $k$ increases, the effect of the advice error in the competitive ratio becomes marginal, even if $H$ is as high as linear in $k$. To see this, we need to be able to evaluate the partial sum of binomial coefficients. Since this partial sum does not have a closed form, we will rely on the following useful approximation from~\cite{macwilliams1977theory}. We have

\begin{equation}
\frac{2^{N{\cal H}(\frac{m}{N})}}{\sqrt{8m(1-\frac{m}{N})}}
\leq \mch{N}{m} \leq 2^{N{\cal H}(\frac{m}{N})}, \quad \textrm{for $0<m<N/2$.} 
\label{eq:amazing}
\end{equation}

Using the above approximation approximation, we have
\begin{align*}
\frac{1}{q}=\frac{\multiset{k}{H}}{2^{k-1}} \leq 
\frac{2^{k{\cal H}(\frac{H}{k})}}
{2^{k-1}}=2^{k({\cal H}(\frac{H}{k})-1)+1}.
\end{align*}

As in the case of Pareto-based analysis, we can interpret the advice as the output of $k$ binary experts, where each expert predicts whether the target lies in an appropriate partition of the infinite line, and up to $H$ experts may give erroneous responses.

\section{Ray search with directional prediction}
\label{sec:direction}

In this section, we study $m$-ray search in the setting in which the prediction is the ray of the hiding target. Without loss of generality, we suppose that the prediction is the ray indexed 0. 

For the upper bound, we consider the following strategy that generalizes the biased search approach of~\cite{demaine:turn}. The searcher fixes some $b>1$ and $\delta>1$, to be specified later, and explores the rays in the cyclic order $0,1, \ldots ,m-1$. If the ray visited in the $i$-th iteration is ray 0, it explores it to a length equal to $\delta b^i$, otherwise, i.e., if the visited ray is in $\{1, \ldots ,m-1\}$, it explores it to a length equal to $b^i$. Thus the search combines elements of geometric search with a bias towards the predicted ray, as expressed by the parameter $\delta$.

\begin{theorem}
For every $b>1$, $\delta>1$, the strategy described above has consistency at most
$
1+2\frac{b^m}{b^m-1}+\frac{2}{\delta} \frac{b^m}{b^m-1} \frac{b^m-b}{b-1},
$
and robustness at most 
$
1+2\delta \frac{b^{m+1}}{b^m-1}+2\frac{b^{m+1}}{b^m-1}\frac{b^m-b}{b-1}-2b^m.
$
\label{thm:upper.2}
\end{theorem}

\begin{proof}
The consistency of the strategy is evaluated for targets hiding on ray 0 and at distances infinitesimally larger than $\delta b^{im}$. The robustness of the strategy, on the other hand, is evaluated for targets hiding on ray $m-1$ and at distances infinitesimally further than $b^{im-1}$. 

Define $x_i=\delta b^{im}$, for all $i \geq 0$ and 
$y_i=\sum_{j=1}^{m-1} b^{im+j}=b^{im+1}\frac{b^{m-1}-1}{b-1}$, that is, $y_i$ is the aggregate length explored in phase $i$, on rays $1, \ldots ,m-1$ . Then we have 
\begin{align*}
\textrm{consistency} &=\sup_i \frac{x_{i}+2(\sum_{j=0}^i x_j+\sum_{j=0}^iy_j)}{x_{i}} \\
&= \sup_i (1+2\frac{b^m-1/b^{im}}{b^m-1}
+\frac{2}{\delta} 
\frac{b^m-b}{b-1} 
\frac{b^m}{b^m-1}) \\
&=1+2\frac{b^m}{b^m-1}+\frac{2}{\delta}\frac{b^m-b}{b-1}\frac{b^m}{b^m-1}.
\end{align*}
Moreover,
\begin{align*}
\textrm{robustness} &= \sup_i \frac{b^{im-1}+2\sum_{j=0}^{i+1} x_j+2\sum_{j=0}^{i+1} y_{j}-b^{im-1}}{b^{im-1}} \\
&= \sup_i \frac{b^{im-1}+2\sum_{j=0}^{i+1} \delta b^j +2\sum_{j=0}^{i+1} b^{jm+1}\frac{b^{m-1}-1}{b-1}
-b^{im-1}} {b^{im-1}},
\end{align*}
which after some calculations yields the result.
\end{proof}

Observe that if $\delta=1$, then both the consistency and the robustness of the above strategy are equal to $1+2\frac{b^m}{b-1}$, as expected (i.e., the competitive ratio of a geometric strategy with base $b$). For any fixed $b$, by increasing $\delta$, the consistency of the resulting strategy improves, at the expense of its robustness. We would like thus to optimize the robustness by choosing $\delta$ as a function of $b$, $m$ and the desired consistency, however it is not obvious that this is possible analytically. Instead, suppose that we choose $b=(m+1)/m$, namely the base of the geometric strategy that results in an optimal competitive ratio for $m$-ray search equal to $r_m^*=1+2\rho_m^*$. Then $\frac{b^m}{b^m-1} \leq e/(e-1)$, which implies that 
the consistency of the strategy is at most
$
1+2\frac{e}{e-1}+\frac{2}{\delta}\frac{e}{e-1}(\rho_m^*-m)=
\frac{2}{\delta}\frac{e}{e-1}(\rho_m^*-m)+O(1).
% 1+\frac{2}{\delta}\frac{e}{e-1}\rho_m^*+2\frac{e}{e-1}(\frac{1}{\delta}-1)
$
% \[
% 1+2\frac{e}{e-1}+\frac{2}{\delta}\frac{e}{e-1}(e-1)m=1+\frac{2}{\delta}em, 
% \]
On the other hand, the robustness of the strategy is at most 
$
1+2\delta\frac{81}{36}+2\frac{81}{36}(\rho_m^*-m)-4= 
\frac{9}{2} (\delta+\rho_m^*-m)+O(1).
$

Therefore, if we would like the strategy to be $c$-consistent, where $c=O(1)+2\tilde{c}$, for some $\tilde{c}$, we can choose $\delta$ to be equal to $\frac{e}{(e-1)\tilde{c}}(\rho_m^*-m)$, and the resulting robustness 
is then at most
$
\frac{9}{2}(\rho_m^*-m)(1+\frac{e}{(e-1)\tilde{c}})+O(1).
$
We can also obtain more precise tradeoffs as $m \to \infty$, since in this case, it is known that
$r_m^*=1+2\rho_m^*=1+2em$. 
\begin{corollary}
For $m \to \infty$, the above strategy has consistency at most $1+2\frac{e}{e-1}+\frac{2}{\delta}em$,
and robustness at most $2e(m+\frac{\delta}{e-1})$. In particular, given $\tilde{c}$, the strategy is 
$(O(1)+2\tilde{c})$-consistent, and $(O(1)+2em(1+\frac{e}{(e-1)\tilde{c}})$-robust. 
\label{cor:infty}
\end{corollary}

% Thus, if we would like to achieve consistency $c=1+2d$, then we must choose $\delta\geq em/d$, which implies robustness at most 
% \[
% 2em(1+\frac{1}{d}\frac{e}{e-1})+1-2e.
% \]

Next, we show a negative result on the tradeoff between the consistency and robustness that any strategy can achieve. The proof follows an approach that we generalize in the weak predictions setting (proof of the lower bound in Theorem~\ref{thm:bound.2.noisy}).

\begin{theorem}
Any $c$-consistent strategy for searching with directional prediction, where $c=1+2\tilde{c}$,
has robustness at least $1+2\rho_{m-1}^*(1+\frac{1}{\tilde{c}-1}))$. In particular, for $m \to \infty$, its
robustness is at least $1+2e(m-1)(1+\frac{1}{\tilde{c}-1})$.
\label{thm:lower.2}
\end{theorem}

\begin{proof}
Any search strategy for the problem operates in {\em phases}, in the sense that the searcher alternates between ray 0 (the predicted ray), and some subset of rays in $\{1, \ldots ,m-1\}$. Namely, every strategy $X$ gives rise to a sequence of the form $(x_i)_{i \geq 0}$, in which $x_i$, for even $i$, describes the explored length of $X$ on ray $0$, and $x_i$ for $i$ odd, describes the aggregate explored length on rays $\{1, \ldots ,m-1\}$. Thus, in phase $i$, the searcher incurs a total cost of $2x_i$, for all $i$ (except for the phase in which the target is found). From the consistency requirement, it must then be that for all $i$,
\[
\frac{x_{2i}+2\sum_{j=0}^{i}x_{2j}+\sum_{j=1}^{i+1} x_{2j-1}}{x_{2i}} \leq c \Rightarrow \frac{\sum_{j=0}^{2i+1}x_j}{x_{2i}} \leq \tilde{c}.
\]
Define $S_i=\sum_{j=0}^i x_{2j}$, and $S_i'=\sum_{j=0}^{i-1}x_{2j+1}$, then the above inequality gives 
\begin{equation}
\frac{S_i+S_i'+x_{2i+1}}{x_{2i}} \leq \tilde{c} \Rightarrow S_i+S'_i \leq \tilde{c}x_{2i} \leq \tilde{c} S_i \Rightarrow \frac{S_i}{S'_i} \geq \frac{1}{\tilde{c}-1}.
\label{eq:lower.2.1}
\end{equation}
In order to bound the robustness of the strategy, consider the phase $2i-1$, in which the searcher explores a subset of rays in $\{1, \ldots ,m-1\}$ with exploration length $x_{2i-1}$. Since searching in a $(m-1)$-ray star has competitive ratio $r_{m-1}^*=1+2\rho_{m-1}^*$, there exists some 
particular value of the index $i$, say $\bar{i}$, and a hiding position, say $t_{\bar i}'$, with the property that
\begin{equation}
1+\frac{2S'_{\bar{i}}}{d(t'_{\bar{i}})} \geq 1+2\rho_{m-1}^*-\epsilon \Rightarrow  \frac{S_{\bar{i}}'}{d'_{\bar{i}}} \geq \rho_{m-1}^*-\epsilon,
\label{eq:lower.2.2}
\end{equation}
where $\epsilon \to 0$, if we allow $i$. To simplify the expressions, we can thus assume that $\epsilon=0$.

Consider the exploration in phase $2i-1$, and a target hiding at $t_{\bar i}'$, as described above. Then the robustness of the strategy is at least
\begin{align}
\frac{2(S_{\bar i}+S_{\bar i}')+d_{\bar i}}{d'_{\bar i}}&=1+2\frac{S_{\bar i}+S_{\bar i}'}{d'_{\bar i}} \notag \\
&\geq 1+2(1+\frac{1}{\tilde{c}-1})\frac{S'_{\bar i}}{d'_{\bar i}}, \tag{From~\eqref{eq:lower.2.1}} \\
& \geq 1+2(1+\frac{1}{\tilde{c}-1})\rho_{m-1}^*. \tag{From~\eqref{eq:lower.2.2}}
\end{align}
\end{proof}

%\medskip

\subsection{Extension to weak predictions}

We consider the model in which if the predicted ray is $h \in \{ 0, \ldots ,m-1 \}$, then the searcher would like to minimize its competitive ratio, assuming that the target hides in one of the rays in the interval $[h-H \bmod m, h+H \bmod m]$, where $H\leq m/2$ is the tolerance parameter. 
This captures the case that the hiding ray is expected to be in a vicinity of the predicted ray, with respect to the tolerance of the searcher. We denote this set of rays by $R_H$, and its complement, by $\overline{R}_H$, and note that $|R_H|=2H+1$ and $|\overline{R}_H|=m-2H-1$. Without loss of generality, we may assume that $R_H=\{0,\ldots 2H\}$.

We prove an asymptotically tight bound on the tradeoff between the competitive ratio and the robustness, that generalizes the error-free setting. Note that every strategy has competitive ratio at least $1+2\rho^*_{2H+1}$, since the weak prediction may incur a search in a $(2H+1)$-ray star. We also obtain an interesting corollary: as the competitive ratio approaches the optimal bound of $1+2\rho^*_{2H+1}$, the robustness increases dramatically. 

\begin{theorem}
For every $\tilde{c}>\rho^*_{2H-1}$ there exists a strategy with directional hint that has competitive ratio $1+2\tilde{c}$, and robustness at most $O(\frac{\rho_{2H+1}^*}{\tilde{c}-\rho^*_{m-2H-1}} (m-2H))$. Furthermore, this bound is tight, i.e., every strategy of competitive ratio $1+2\tilde{c}$
has robustness at least $\Omega(\frac{\rho_{2H+1}^*}{\tilde{c}-\rho^*_{m-2H-1}} (m-2H))$.
\label{thm:bound.2.noisy}
\end{theorem}

\begin{proof}
We first prove the upper bound, which generalizes the strategy we used in the context of consistency/robustness tradeoffs. Define $b=\frac{2H+1}{2H}$, i.e., the optimal base of a geometric strategy for searching in a $(2H+1)$-ray star, and $\delta>1$, to be specified later. Consider a cyclic strategy 
which visits rays $0, \ldots ,m-1$ in this order, and which works in rounds. Specifically, in round $i$, it explores ray $j \in R_{H}$ to length $\delta b^{(2H+1)i+j}$, and every ray in $\overline{R}_H$ to length $b^{(2H+1)i+2H}$.
The competitive ratio of this strategy, assuming error at most $H$, is maximized for targets hiding infinitesimally beyond the turn points on ray $2H$ in $R_H$. Simple calculations show that the competitive ratio is
\[
1+2\rho^*_{2H+1}+\Theta(\frac{1}{\delta} \frac{b^{2H+1}}{b^{2H+1}-1} (m-2H-1)),
\]
and note that $\frac{b^{2H+1}}{b^{2H+1}-1}$ is at most $e/(e-1)$, by the choice of $b$. 
Thus, for the competitive ratio to be at most $1+2\tilde{c}$, it must be that
\begin{equation}
\delta \in \Omega(\frac{m-2H}{\tilde{c}-\rho_{2H+1}^*}).
\label{eq:delta.noisy}
\end{equation}
The robustness of the strategy is evaluated for a target hiding at distance infinitesimally beyond the turn points of the searcher on ray $m-1$. After simple calculations we obtain that the robustness is at most $O(\delta \rho_{2H+1}^*(m-2H+1))$, which from~\eqref{eq:delta.noisy} is at most $O(\frac{\rho_{2H+1}^*}{\tilde{c}-\rho^*_{2H+1}} (m-2H))$, and which proves the upper bound. 

We now proceed with the lower bound. Any strategy for the problem consists of {\em phases}, which alternate between searching a subset of  $R_H$ and a subset of $\overline{R}_H$. Namely, every strategy $X$ is of the form $X=(x_i)_{i \geq 0}$, in which $x_i$, for even $i$, describes the aggregate explored length of $X$ on rays that belong exclusively in  $R_H$, and $x_i$ for $i$ odd, describes the aggregate explored length on rays that belong exclusively in $\overline{R}_H$. Thus, in each phase $i$, the searcher incurs a cost of $2x_i$, for all $i$ (except for the phase at which the target is found). 

Define $S_i=\sum_{j=0}^i x_{2j}$, and $S_i'=\sum_{j=0}^{i-1}x_{2j+1}$. Then from the competitiveness of the strategy, for all even $i$, and for any target $t_H$ hiding in $R_H$ that is discovered in phase $2(i+1)$, 
\vspace{-0.2cm}
\begin{equation}
1+2\frac{S_i+S'_i}{d(t_H)} \leq c \Rightarrow \tilde{c} \geq \frac{S_i+S'_i}{d(t_H)}.
\label{eq:noisy.lower.direction.1}
\end{equation}
We know that searching in a $(2H+1)$-ray has competitive ratio at least $1+2\rho^*_{2H+1}$. This means that for all $i$, there exists a target $d_i$ that is discovered in phase $2(i+1)$ such that 
\[
1+2\sup_i \frac{S_i}{d_i} \geq 1+2\rho^*_{2H+1}, 
\]
hence there exists some $\bar{i}$ for which the above inequality gives
$1+2\frac{S_{\bar{i}}}{d_{\bar{i}}} \geq 1+2\rho^*_{2H+1} -\epsilon$, 
where $\epsilon \to 0$, as $\bar{i}$ is allowed to be unbounded. To simplify the exposition, we can thus assume that $\epsilon=0$, and obtain
\begin{equation}
d(t_{\bar{i}}) \leq \frac{S_{\bar{i}}}{\rho^*_{2H+1}}.
\label{eq:noisy.lower.direction.2}
\end{equation}
Moreover, for any $i$, there exists a hiding position for a target $t'_i$ in $\overline{R}_H$ that is first discovered in phase $2i+1$ it must be that
\begin{equation}
1+2\frac{S_i'}{d(t_i')} \geq m-2H-1 \Rightarrow d(t_i') =O(\frac{S_i'}{m-2H-1}),
\label{eq:noisy.lower.direction.3}
\end{equation}
since $2S_i'$ describes the cost incurred by the searcher on rays in $\overline{R}_H$, right before phase $2i+1$ starts. Note also that this inequality holds {\em for all i}, unlike~\eqref{eq:noisy.lower.direction.2}, that holds only for $\bar{i}$.

To bound the robustness, consider the phase $2\bar{i}+1$, with $\bar{i}$ as defined above, and the target $t_{\bar{i}}'$, again as defined above. Then we have that
\begin{equation}
\textrm{Robustness} \geq 1+2\frac{S_{\bar{i}}+S'_{\bar{i}}}{d(t_{\bar{i}}')} 
=\Omega((1+\frac{S_{\bar{i}}}{S'_{\bar{i}}}) (m-2H)), 
\label{eq:noisy.lower.direction.4}
\end{equation}
from~\eqref{eq:noisy.lower.direction.3}.
Moreover, from~\eqref{eq:noisy.lower.direction.1} and~\eqref{eq:noisy.lower.direction.2} we have that
\[
\frac{S_{\bar{i}}+S'_{\bar{i}}}{S'_{\bar{i}}} \leq \frac{\tilde{c}} {\rho^*_{2H+1}} \Rightarrow
\frac{S_{\bar{i}}}{S'_{\bar{i}}} \geq \frac{\rho^*_{2H+1}}{\tilde{c}-\rho^*_{2H+1}},
\]
and substituting the above inequality to~\eqref{eq:noisy.lower.direction.1} yields the result.
\end{proof}

\section{Ray search with positional prediction}
\label{sec:position}

In this section we study $m$-ray searching in the setting in which the prediction is the position of the target in the star environment. Namely, the prediction $h$ is a pair $(d_h,u_h)$, where $d_h$ is the predicted distance from $O$ and $u_h$ is the predicted ray. 
We first show the upper bound.
 % that generalizes the {\em scaled-aggressive} strategy of~\cite{DBLP:conf/innovations/000121} to the star environment. 
\begin{theorem}
For any $r\geq r_m^\star$, there is an $r$-robust strategy of consistency at most $1+2\frac{1}{b_r-1}.$
\label{thm:3.upper}
\end{theorem}

\begin{proof}
Consider the cyclic, geometric strategy $G_{b_r}=(b_r^i)_{i \geq 0}$, and let $j_h$ be an index such that  $b_r^{j_h-m} < d_h \leq b_r^{j_h}$. Define $\lambda=b_r^{j_h}/d_h \geq 1$,  and let $G'$ denote the cyclic geometric strategy $G'=(\frac{1}{\lambda} b_r^i)_{i \geq 0}$,
which explores the rays in the same cyclic order as $G$. Note that since $G_{b_r}$ is $r$-robust, so is the scaled-down strategy $G'$. It remains then to bound the
consistency of $G'$. We have that 
\[
d(G',h)=\frac{1}{\lambda}(2\sum_{i=0}^{j_h-1} b_r^i +b_r^{j_h}),
\]
and since $d_h=b_r^{j_h}/\lambda$, the consistency of the strategy is at most 
\[
\frac{d(G',h)}{d_h} \leq 1+2 \frac{b_r^{j_h}-1}{b_r^{j_h}(b_r-1)}\leq 1+2\frac{1}{b_r-1}.
\]
\end{proof}

We will now show that the strategy of Theorem~\ref{thm:3.upper} is Pareto-optimal. The proof of the following theorem generalizes, but also simplifies the lower bound of~\cite{DBLP:conf/innovations/000121} which applies only to linear search ($m=2$). The crux in the proof is to exploit the properties of the parameter $\alpha_Y$, where $Y$ will be defined as the sequence of search lengths of a cyclic strategy defined by a linear recurrence relation. In particular, these properties allow us to bypass technical complications related to the study of such relations, by establishing appropriate lower bounds (as opposed to solving the recurrence relation).

\begin{theorem}
For any $r\geq r_m^\star$, no $r$-robust strategy for searching with positional prediction has consistency better than $1+2\frac{1}{b_r-1}.$
\label{thm:3.lower}
\end{theorem}

\begin{proof}
 Let $X$ denote an $r$-robust strategy. From previous studies of star search in~\cite{ultimate,jaillet:online}, we know that any $r$-competitive search strategy $X$ can be transformed to a {\em cyclic}, $r$-competitive strategy  which has the same search lengths as $X$, and in iteration $i$ 
explores the ray to a length equal to the $i$-th smallest search length in $X$. 
Since this strategy is cyclic, its competitive ratio is given by~\eqref{eq:comp.cyclic}. Furthermore, using the same argument as in~\cite{ultimate} (namely, Lemma 2.2 in~\cite{ultimate}), there exists a cyclic strategy $Y$ that minimizes the consistency $\sup_h \frac{d(Y,h)}{d(h)}$, in which the competitiveness constraints~\eqref{eq:comp.cyclic} are satisfied with equality for iterations up to the discovery of the target. Formally, there exists a cyclic strategy $Y=(y_i)_i$ in which
\[
y_{j_h}=h, \text{for some index $j_h$ } \quad \text{and } \quad 1+2\frac{\sum_{i=0}^j y_i}{y_{j-m+1}}=r, \text{for all $j \in [m-1,j_h]$.}
\]
It follows that for all $j \leq j_h$, the search length $y_j$ is determined by the recurrence 
\begin{equation}
y_j=r(y_{j-m+1}-y_{j-m}), 
\label{eq:linear.rec}
\end{equation}
with some initial conditions $y_0, \ldots ,y_{m-1}$. We have that
\begin{equation}
d(Y,h)= \sum_{j=0}^{j_h} y_j =d_h+2\sum_{j=0}^{j_h-1} y_j=d_h+2r(y_{j_h-m+1}-y_0)+\sum_{j=0}^{m-1}y_j,
\label{eq:3.lower.1}
\end{equation}
where we used the fact that $\sum_j{y_j}$ is telescoping, as seen by~\eqref{eq:linear.rec}.
From~\eqref{eq:3.lower.1} we have 
$
\frac{d(Y,h)}{d_h} = 
1+2\frac{ry_{j_h-m+1}}{d_h}+\frac{1}{d_h}(\sum_{j=0}^{m-1}y_j-y_0).
$
% For a target $h$ at distance $d_h \to \infty$, we have that $\frac{1}{d_h}(\sum_{j=0}^{m-1}y_j-y_0) \to 0$. Since $Y$ is cyclic, we obtain
Since $Y$ is cyclic, and $d_h=y_{j_h}$, we obtain that
\begin{equation}
\sup_h \frac{d(Y,h)}{d_h} \geq 1+2 \sup_{j_h} \frac{ry_{j_h-m+1}}{y_{j_h}} =1+2r \sup_{j_h}\frac{y_{j_h-m+1}}{y_{j_h}}.
\label{eq:3.lower.2}
\end{equation}
Define the functional $F_j(Y)=\frac{y_{j-m+1}}{y_j}$. This functional satisfies the conditions of Theorem~\ref{thm:sch}, hence 
$
\sup_j \frac{y_{j-m+1}}{y_j} \geq \frac{\alpha_Y^{j-m+1}}{\alpha_Y}=\alpha_Y^{-m}.  
$
Since $Y$ is $r$-robust, as discussed in Section~\ref{sec:preliminaries}
$
r \geq \comp(X) \geq \frac{\alpha_Y^m}{\alpha_Y-1},
$
where it must be $\alpha_Y \leq b_r$, from the definition of $b_r$. Thus,~\eqref{eq:3.lower.2} gives
\[
\sup_h \frac{d(Y,h)}{d_h} \geq 1+2 \frac{\alpha_Y^m}{\alpha_Y-1} \alpha_Y^{-m}=1+2\frac{1}{\alpha_Y-1}
\geq 1+2\frac{1}{b_r-1}.
\]
Hence, the consistency of $Y$ is at least $1+2\frac{1}{b_r-1}$, and thus so is the consistency of $X$.
\end{proof}

\subsection{Extension to weak predictions}

Given the prediction $h$ related to a target $t$, we define the prediction {\em error} as the distance between $t$ and $h$ in the star, normalized by the distance $d(h)$ of the prediction's position from the origin. Namely, 
$
\eta=\frac{|d(t)-d(h)|}{d(h)}.
$
We distinguish between different types of the error: If $d$ and $h$ are in the same ray, but $d(t)>d(h)$ we call the error {\em positive}, whereas if $d$ and $h$ are in the same ray, but $d(t)<d(h)$ we call the error {\em negative}. If the error is neither positive or negative, then $h$ and $t$ are in different rays.

Let $X_h$ denote the Pareto-optimal strategy of Theorem~\ref{thm:3.upper}. Let $H>0$ denote the tolerance parameter that is specified by the searcher, and consider the strategy $X_{h(1+H)}$, i.e., the strategy that pretends that the prediction is at the same ray as $h$, but at distance $d(h)(1+H)$ from $O$. The following result is a corollary of Theorem~\ref{thm:3.upper}.

\begin{corollary}
For any $H>0$ and $r\geq r_m^*$, strategy $X_{h(1+H)}$ is $r$-robust and 
% has consistency at most $\min\{1+2\frac{1+H}{b_r-1},r\}$.
has competitive ratio at most $\min \{1+2\frac{1+H}{b_r-1} ,r\}$, if the error is either positive or negative, and at most $H$. Otherwise, its competitive ratio is at most $r$.
\label{thm:3.noisy.upper}
\end{corollary}

Last, we show that the above tradeoffs are tight.
\begin{theorem}
For any $r$-robust strategy with positional prediction, there exists $q>0$ such that its competitive ratio is no better than $\min\{1+2\frac{1+q}{b_r-1},r\}$ for positive or negative error at most $q$. Otherwise, its competitive ratio is at least $r$. 
\label{thm:3.noisy.lower}
\end{theorem}

\begin{proof}
Let $X_h$ denote any $r$-robust strategy given prediction $h$. Consider first the case of positive error. Let $i_h$ be the iteration at which 
$X_h$ discovers $h$, then there must exist $p \in  [0,1]$ such that $x_{i_h} \geq (1+p)d_h$. Consider first the case $p>0$. Then from Theorem~\ref{thm:3.noisy.lower}, there exists $h$ (such that $d_h$ is sufficiently large), for which the consistency of $X_h$ is no better than $1+2\frac{1}{b_r-1}$. 
Let $\tilde{h}$ denote a position on the same ray as $h$, with $d(\tilde{h})=d_h(1+p)$. 
The above statement  implies that 
\[
d(X_h,\tilde{h})\geq d(\tilde{h})+2\frac{(1+p)d_h}{b_r-1},
\] 
and therefore for every target $t$ hiding on the same ray as $h$, and such that $d(t)=(1+q)d_h$, with 
$0< q\leq \frac{p}{4} $, we have 
\begin{align*}
\frac{d(X_h,t)}{d(t)} &\geq 1+2\frac{1+p}{(1+q)b_r-1} \\ &\geq
1+2\frac{1+p}{(1+\frac{p}{4})b_r-1} \\ 
&\geq
1+2\frac{1+\frac{p}{4}}{b_r-1} \\
&\geq 1+2\frac{1+q}{b_r-1}.
\end{align*}
 If $p=0$, then consider $t$ hiding at distance $(1+\epsilon)d_h$, and in the same ray as $h$. In this case, $X_h$ would locate $t$ in iteration $i_{h+1}$
at the earliest, and thus in this case the competitive ratio of $X_h$ is at least 
\[
3+\frac{2}{b_1-1}>1+\frac{2(1+\epsilon)}{b_r-1},
\]
for sufficiently small $\epsilon$. This settles the case that the error is positive.

Consider now the case that the error is negative, i.e., $d(t)=(1-q)d_h$. From Theorem~\ref{thm:3.noisy.lower} we have that 
\[
d(X,t)\geq d(t)+2\frac{1}{(1-q)(b_r-1)}, 
\]
therefore,
\[
\frac{d(X_h,t)}{d(t)} \geq 1+2\frac{1}{(1-q)b_r-1} \geq 1+2\frac{1+q}{b_r-1},
\]
for sufficiently small $q$.

Last, in all other cases, the competitive ratio of $X_h$ can be as high as its robustness, since the position of the target may be chosen to be any point in the star environment that maximizes the competitive ratio of the strategy. 
\end{proof}

\bibliographystyle{plain}
\bibliography{fault-contract,targets-arxiv,targets,refs}

\begin{thebibliography}{10}

\bibitem{searchgames}
Steve Alpern and Shmuel Gal.
\newblock {\em The theory of search games and rendezvous}, volume~55 of {\em
  International series in operations research and management science}.
\newblock Kluwer, 2003.

\bibitem{AL:expanding}
Steve Alpern and Thomas Lidbetter.
\newblock Mining coal or finding terrorists: The expanding search paradigm.
\newblock {\em Operations Research}, 61(2):265--279, 2013.

\bibitem{spyros:rays}
Spyros Angelopoulos.
\newblock Further connections between contract-scheduling and ray-searching
  problems.
\newblock In {\em Proceedings of the 24th International Joint Conference in
  Artificial Intelligence (IJCAI)}, pages 1516--1522, 2015.

\bibitem{DBLP:conf/innovations/000121}
Spyros Angelopoulos.
\newblock Online search with a hint.
\newblock In {\em Proceedings of the 12th Innovations in Theoretical Computer
  Science Conference {(ITCS)}}, pages 51:1--51:16, 2021.

\bibitem{DBLP:conf/stacs/0001DJ19}
Spyros Angelopoulos, Christoph D{\"{u}}rr, and Shendan Jin.
\newblock Best-of-two-worlds analysis of online search.
\newblock In {\em Proceedings of the 36th {International} {Symposium} on
  {Theoretical} {Aspects} of {Computer} {Science} ({STACS})}, pages 7:1--7:17,
  2019.

\bibitem{stacs-expanding}
Spyros Angelopoulos, Christoph D\"urr, and Thomas Lidbetter.
\newblock The expanding search ratio of a graph.
\newblock In {\em Proceedings of the 33rd Symposium on Theoretical Aspects of
  Computer Science (STACS)}, pages 9:1--9:14, 2016.

\bibitem{DBLP:conf/aaai/0001K21}
Spyros Angelopoulos and Shahin Kamali.
\newblock Contract scheduling with predictions.
\newblock In {\em Proceedings of the 35th {AAAI} Conference on Artificial
  Intelligence}, pages 11726--11733, 2021.

\bibitem{yates:plane}
Ricardo~A. Baeza-Yates, Joseph~C. Culberson, and Gregory~G.E. Rawlins.
\newblock Searching in the plane.
\newblock {\em Information and Computation}, 106:234--244, 1993.

\bibitem{DBLP:conf/innovations/BanerjeeC0L23}
Siddhartha Banerjee, Vincent Cohen{-}Addad, Anupam Gupta, and Zhouzi Li.
\newblock Graph searching with predictions.
\newblock In {\em {Proceedings of the 14th Conference on Innovations in
  Theoretical Computer Science (ITCS)}}, volume 251 of {\em LIPIcs}, pages
  12:1--12:24, 2023.

\bibitem{beck:ls}
Anatole Beck.
\newblock On the linear search problem.
\newblock {\em Naval Research Logistics}, 2:221--228, 1964.

\bibitem{beck:yet.more}
Anatole Beck and Donald~J. Newman.
\newblock Yet more on the linear search problem.
\newblock {\em Israel Journal of Mathematics}, 8:419--429, 1970.

\bibitem{bellman}
Richard Bellman.
\newblock An optimal search problem.
\newblock {\em SIAM Review}, 5:274, 1963.

\bibitem{steins}
Daniel~S. Bernstein, Lev Finkelstein, and Shlomo Zilberstein.
\newblock Contract algorithms and robots on rays: Unifying two scheduling
  problems.
\newblock In {\em Proceedings of the 18th International Joint Conference on
  Artificial Intelligence (IJCAI)}, pages 1211--1217, 2003.

\bibitem{boczkowski2018searching}
Lucas Boczkowski, Amos Korman, and Yoav Rodeh.
\newblock Searching a tree with permanently noisy advice.
\newblock In {\em ESA 2018-26th Annual European Symposium on Algorithms}, pages
  1--32, 2018.

\bibitem{DBLP:conf/latin/BonatoGMP20}
Anthony Bonato, Konstantinos Georgiou, Calum MacRury, and Pawel Pralat.
\newblock Probabilistically faulty searching on a half-line - (extended
  abstract).
\newblock In {\em {LATIN} 2020: Theoretical Informatics - 14th Latin American
  Symposium, S{\~{a}}o Paulo, Brazils}, volume 12118 of {\em Lecture Notes in
  Computer Science}, pages 168--180. Springer, 2020.

\bibitem{revisiting:esa}
Prosenjit Bose, Jean-Lou~De Carufel, and Stephane Durocher.
\newblock Searching on a line: A complete characterization of the optimal
  solution.
\newblock {\em Theoretical Computer Science}, 569:24--42, 2015.

\bibitem{eleos}
Marek Chrobak, Leszek Gasieniec, Thomas Gorry, and Russell Martin.
\newblock Group search on the line.
\newblock In {\em International Conference on Current Trends in Theory and
  Practice of Informatics}, pages 164--176. Springer, 2015.

\bibitem{Condon:2009:ADA:1497290.1497300}
Anne Condon, Amol Deshpande, Lisa Hellerstein, and Ning Wu.
\newblock Algorithms for distributional and adversarial pipelined filter
  ordering problems.
\newblock {\em ACM Transaction on Algorithms}, 5(2):24:1--24:34, 2009.

\bibitem{DBLP:journals/dc/CzyzowiczKKNO19}
Jurek Czyzowicz, Evangelos Kranakis, Danny Krizanc, Lata Narayanan, and
  Jaroslav Opatrny.
\newblock Search on a line with faulty robots.
\newblock {\em Distributed Comput.}, 32(6):493--504, 2019.

\bibitem{demaine:turn}
Erik~D Demaine, S{\'a}ndor~P Fekete, and Shmuel Gal.
\newblock Online searching with turn cost.
\newblock {\em Theoretical Computer Science}, 361:342--355, 2006.

\bibitem{dobrev2012online}
Stefan Dobrev, Rastislav Kr{\'a}lovi{\v{c}}, and Euripides Markou.
\newblock Online graph exploration with advice.
\newblock In {\em Structural Information and Communication Complexity: 19th
  International Colloquium, SIROCCO 2012}, pages 267--278. Springer, 2012.

\bibitem{eberle2022robustification}
Franziska Eberle, Alexander Lindermayr, Nicole Megow, Lukas N{\"o}lke, and Jens
  Schl{\"o}ter.
\newblock Robustification of online graph exploration methods.
\newblock In {\em Proceedings of the AAAI Conference on Artificial
  Intelligence}, volume~36, pages 9732--9740, 2022.

\bibitem{fraigniaud2008tree}
Pierre Fraigniaud, David Ilcinkas, and Andrzej Pelc.
\newblock Tree exploration with advice.
\newblock {\em Information and Computation}, 206(11):1276--1287, 2008.

\bibitem{gal:general}
Shmuel Gal.
\newblock A general search game.
\newblock {\em Israel Journal of Mathematics}, 12:32--45, 1972.

\bibitem{gal:minimax}
Shmuel Gal.
\newblock Minimax solutions for linear search problems.
\newblock {\em SIAM Journal on Applied Mathematics}, 27:17--30, 1974.

\bibitem{Gal80}
Shmuel Gal.
\newblock {\em Search Games}.
\newblock Academic Press, 1980.

\bibitem{DBLP:journals/csr/GhoshK10}
Subir~Kumar Ghosh and Rolf Klein.
\newblock Online algorithms for searching and exploration in the plane.
\newblock {\em Comput. Sci. Rev.}, 4(4):189--201, 2010.

\bibitem{gorain2018deterministic}
Barun Gorain and Andrzej Pelc.
\newblock Deterministic graph exploration with advice.
\newblock {\em ACM Transactions on Algorithms}, 15(1):1--17, 2018.

\bibitem{HIKL99:fixed.distance}
C.~Hipke, C.~Icking, R.~Klein, and E.~Langetepe.
\newblock How to find a point in the line within a fixed distance.
\newblock {\em Discrete Applied Mathematics}, 93:67--73, 1999.

\bibitem{knapsack:frequency}
Sungjin Im, Ravi Kumar, Mahshid~Montazer Qaem, and Manish Purohit.
\newblock Online knapsack with frequency predictions.
\newblock In {\em Proceedings of the 34th Annual Conference on Neural
  Information Processing Systems ({NeurIPS})}, pages 2733--2743, 2021.

\bibitem{jaillet:online}
Parick Jaillet and Matthew Stafford.
\newblock Online searching.
\newblock {\em Operations Research}, 49:234--244, 1993.

\bibitem{hybrid}
Ming-Yang Kao, Yuan Ma, Michael Sipser, and Yiqun Yin.
\newblock Optimal constructions of hybrid algorithms.
\newblock {\em Journal of Algorithms}, 29(1):142--164, 1998.

\bibitem{ray:2randomized}
Ming-Yang Kao, John~H Reif, and Stephen~R Tate.
\newblock Searching in an unknown environment: an optimal randomized algorithm
  for the cow-path problem.
\newblock {\em Information and Computation}, 131(1):63--80, 1996.

\bibitem{hyperbolic}
David~G. Kirkpatrick.
\newblock Hyperbolic dovetailing.
\newblock In {\em Proceedings of the 17th Annual European Symposium on
  Algorithms (ESA)}, pages 616--627, 2009.

\bibitem{komm2015treasure}
Dennis Komm, Rastislav Kr{\'a}lovi{\v{c}}, Richard Kr{\'a}lovi{\v{c}}, and
  Jasmin Smula.
\newblock Treasure hunt with advice.
\newblock In {\em Structural Information and Communication Complexity: 22nd
  International Colloquium, SIROCCO 2015}, pages 328--341. Springer, 2015.

\bibitem{koutsoupias:fixed}
E.~Koutsoupias, C.H. Papadimitriou, and M.~Yannakakis.
\newblock Searching a fixed graph.
\newblock In {\em Proc. of the 23rd Int. Colloq. on Automata, Languages and
  Programming (ICALP)}, pages 280--289, 1996.

\bibitem{kupavskii2018lower}
Andrey Kupavskii and Emo Welzl.
\newblock Lower bounds for searching robots, some faulty.
\newblock In {\em Proceedings of the 37th ACM Symposium on Principles of
  Distributed Computing {(PODC)}}, pages 447--453, 2018.

\bibitem{langetepe2010optimality}
Elmar Langetepe.
\newblock On the optimality of spiral search.
\newblock In {\em Proceedings of the Twenty-First Annual ACM-SIAM Symposium on
  Discrete Algorithms}, pages 1--12. SIAM, 2010.

\bibitem{lee2021online}
Russell Lee, Jessica Maghakian, Mohammad~H. Hajiesmaili, Jian Li, Ramesh~K.
  Sitaraman, and Zhenhua Liu.
\newblock Online peak-aware energy scheduling with untrusted advice.
\newblock In {\em Proceedings of the 12th {ACM} International Conference on
  Future Energy Systems (eEnergy)}, pages 107--123. {ACM}, 2021.

\bibitem{li2021robustness}
Tongxin Li, Ruixiao Yang, Guannan Qu, Guanya Shi, Chenkai Yu, Adam Wierman, and
  Steven~H. Low.
\newblock Robustness and consistency in linear quadratic control with untrusted
  predictions.
\newblock {\em Proc. {ACM} Meas. Anal. Comput. Syst.}, 6(1):18:1--18:35, 2022.

\bibitem{DBLP:journals/eor/Lidbetter20}
Thomas Lidbetter.
\newblock Search and rescue in the face of uncertain threats.
\newblock {\em Eur. J. Oper. Res.}, 285(3):1153--1160, 2020.

\bibitem{predictionslist}
Alexander Lindermayr and Nicole Megow.
\newblock Repository of works on algorithms with predictions.
\newblock \url{https://algorithms-with-predictions.github.io/about}, 2023.
\newblock Accessed: 2023-04-01.

\bibitem{ultimate}
Alejandro L\'opez-Ortiz and Svem Schuierer.
\newblock The ultimate strategy to search on $m$ rays?
\newblock {\em Theoretical Computer Science}, 261(2):267--295, 2001.

\bibitem{alex:robots}
Alejandro L\'opez-Ortiz and Sven Schuierer.
\newblock On-line parallel heuristics, processor scheduling and robot searching
  under the competitive framework.
\newblock {\em Theoretical Computer Science}, 310(1--3):527--537, 2004.

\bibitem{DBLP:conf/icml/LykourisV18}
Thodoris Lykouris and Sergei Vassilvitskii.
\newblock Competitive caching with machine learned advice.
\newblock {\em J. {ACM}}, 68(4):24:1--24:25, 2021.

\bibitem{macwilliams1977theory}
Florence~Jessie MacWilliams and Neil James~Alexander Sloane.
\newblock {\em The theory of error correcting codes}, volume~16.
\newblock Elsevier, 1977.

\bibitem{oil}
Andrew McGregor, Krzysztof Onak, and Rina Panigrahy.
\newblock The oil searching problem.
\newblock In {\em European Symposium on Algorithms}, pages 504--515, 2009.

\bibitem{DBLP:books/cu/20/MitzenmacherV20}
Michael Mitzenmacher and Sergei Vassilvitskii.
\newblock Algorithms with predictions.
\newblock In {\em Beyond the Worst-Case Analysis of Algorithms}, pages
  646--662. Cambridge University Press, 2020.

\bibitem{pelc2021advice}
Andrzej Pelc and Ram~Narayan Yadav.
\newblock Advice complexity of treasure hunt in geometric terrains.
\newblock {\em Information and Computation}, 281:104705, 2021.

\bibitem{NIPS2018_8174}
Manish Purohit, Zoya Svitkina, and Ravi Kumar.
\newblock Improving online algorithms via {ML} predictions.
\newblock In {\em Proceedings of the 32nd Conference on Neural Information
  Processing Systems ({NeurIPS})}, pages 9661--9670, 2018.

\bibitem{RivestMKWS80}
Ronald~L. Rivest, Albert~R. Meyer, Daniel~J. Kleitman, Karl Winklmann, and Joel
  Spencer.
\newblock Coping with errors in binary search procedures.
\newblock {\em J. Comput. Syst. Sci.}, 20(3):396--404, 1980.

\bibitem{schuierer:lower}
Sven Schuierer.
\newblock Lower bounds in online geometric searching.
\newblock {\em Computational Geometry: Theory and Applications}, 18(1):37--53,
  2001.

\bibitem{wei2020optimal}
Alexander Wei and Fred Zhang.
\newblock Optimal robustness-consistency trade-offs for learning-augmented
  online algorithms.
\newblock In {\em Proceedings of the 33rd Conference on Neural Information
  Processing Systems ({NeurIPS})}, 2020.

\end{thebibliography}

\end{document}